\documentclass[12pt]{amsart}
\usepackage{amsmath,amssymb}
\usepackage{url,verbatim,comment}
\usepackage[usenames]{color}
\usepackage{hyperref}
\input{glyphtounicode}  % These two lines solve the problem of searching
\newtheorem{theorem}{Theorem}%[section]
\newtheorem{claim}[theorem]{Claim}
\newtheorem{corollary}[theorem]{Corollary}
\newtheorem{lemma}[theorem]{Lemma}
\newtheorem{proposition}[theorem]{Proposition}

\theoremstyle{definition}
\newtheorem{definition}[theorem]{Definition}

\theoremstyle{remark}
\newtheorem{example}{Example}

\newcommand{\noprint}[1]{\relax}

\newcommand{\bra}[1]{\langle#1|}
\newcommand{\ket}[1]{|#1\rangle}
\newcommand{\ketpsi}{\ket{\psi}}

\newcommand{\C}{\mathbb C}
\newcommand{\CP}{\mathbb{CP}}
\newcommand{\dom}[1]{\text{Dom}(#1)}
\newcommand{\E}{\mathbb E}
\newcommand{\e}{\vec{e}}
\newcommand{\eps}{\varepsilon}

\newcommand{\FS}{\text{FS}}

\renewcommand{\H}{\mathcal H}

\renewcommand{\L}{\mathcal L}

\newcommand{\Norm}[1]{\Vert#1\Vert}
\renewcommand{\P}{\mathcal P}

\newcommand{\PL}{\text{PL}}
\newcommand{\Q}{Q}

\newcommand{\QL}{\Q L}
\newcommand{\QM}{\Q M}
\newcommand{\range}[1]{\text{Range}(#1)}

\newcommand{\s}{\mathcal S}
\newcommand{\Set}{\text{Set}}
\newcommand{\sq}[1]{\ensuremath{\langle#1\rangle}}
\newcommand{\SP}{\text{SP}}

\renewcommand{\th}{\ensuremath{{}^{\text{th}}}}
\newcommand{\Tube}{\text{Tube}\,}
\newcommand{\U}{\hat{U}}
\renewcommand{\u}{\vec{u}}
\renewcommand{\v}{\vec{v}}
\newcommand{\w}{\vec{w}}

\title[Quantum state transformations]
{Ancilla approximable\\ quantum state transformations}
\author{Andreas Blass}
\address{Mathematics Department\\
University of Michigan\\
Ann Arbor, MI 48109--1043, U.S.A.}
\email{ablass@umich.edu}
\author{Yuri Gurevich}
\address{Microsoft Research\\
Redmond, WA 98052, U.S.A.}
\email{gurevich@microsoft.com}

\begin{document}

\begin{abstract}
We consider the transformations of quantum states obtainable by a process of the following sort.  Combine the given input state with a specially prepared initial state of an auxiliary system.  Apply a unitary transformation to the combined system.  Measure the state of the auxiliary subsystem.  If (and only if) it is in a specified final state, consider the process successful, and take the resulting state of the original (principal) system as the result of the process.

We review known information about exact realization of transformations by such a process.  Then we present results about approximate realization of finite partial transformations.  We consider primarily the issue of approximation to within a specified positive $\eps$, but we also address the question of arbitrarily close approximation.
\end{abstract}

\maketitle

\section{Introduction and main results} %1
\label{sec:intro}

Consider an experiment involving the composition of two
distinguishable quantum systems, a principal and an auxiliary
one. Initially the auxiliary system is in a prepared initial state,
and the principal system is in an arbitrary state $\ketpsi$. Apply a
unitary operator $U$ to the composite system, then  measure the
auxiliary system, and declare success if the auxiliary system is found
to be in a particular (designated a priori) final state. In the case
of success, let $\U\ketpsi$ be the resulting state of the principal
system. The transformation $\U$ is not necessarily unitary or even
total.

Such an experiment is a recurring theme in recent quantum-computation
literature; see  \cite{BSC,Bogus,CGMSY,CW,DS,Jetal,NC,WK} for
example. Typically one tries to maximize the probability that the
measurement is successful and the state $\U\ketpsi$ is of some desired
form, and one may or may not be able to use the resulting state of the
principal system if the measurement is not successful. In particular,
Childs and Wiebe use such an experiment to simulate convex linear
combinations of unitary operators \cite{CW}.

A number of natural questions arise including these:
\begin{itemize}
\item Which state transformations of the principal system can be
  exactly realized that way?
\item What success probability can be guaranteed?
\item How many ancillas are needed to achieve the desired results?
\end{itemize}
Much depends of course on the constraints imposed on the unitary
operator $U$. In the simple case where no restrictions are placed on
$U$, the answers to the three questions are known. We summarize them
in the following Exact Realization Theorem. But first we need a few
definitions.

Let $\H, \H^+$ be the Hilbert spaces for the principal and composite
systems respectively. We presume that $\H^+$ is finite dimensional. If
$\ket{\alpha_1}$ and $\ket{\alpha_2}$ are the designated initial and
final states of the auxiliary system and if the measurement is
successful, then
\begin{equation*}
 \U\ketpsi = \pi_0 U(\ket{\alpha_1}\otimes\ketpsi)
\end{equation*}
where $\pi_0$ is the composition
\[
 \H^+ \to \{\ket{\alpha_2}\}\otimes\H \to \H
\]
of a projection and an isomorphism, and the vector $\U\ketpsi$ is
unnormalized. As in much of the literature, we usually ignore this distinction between a nonzero vector in a Hilbert space and the state represented by the vector, though we try to pay attention to the distinction in formal definitions and theorems.
The following definition takes into account that nonzero
collinear state vectors  represent the same state. (Two vectors are
collinear if one of them is a nonzero multiple of the other.)

\begin{definition}[Exact Realization]
A unitary operator $U$ on $\H^+$ \emph{exactly realizes} a partial
transformation $T$ of $\H - \{\vec{0}\}$ (into itself) if $\U\ketpsi$
is nonzero and collinear with $T\ketpsi$ for every $\ketpsi$ in the
domain $\dom{T}$ of $T$.
\end{definition}

The success probability $\SP(U,\ketpsi)$ of $U$ on a normalized $\H$
vector $\ketpsi$ is $\Norm{\U\ketpsi}^2$. The \emph{guaranteed success
  probability} of $U$ is
$$
 \min \{ \SP(U,\ketpsi):\; \Norm{\ketpsi} = 1 \}.
$$
The use of min (rather than inf) is justified because the space of
unitary operators is compact. Every linear operator on $\H$ can be
viewed as a partial transformation\footnotemark of $\H -
\{\vec{0}\}$.
\footnotetext{By ``partial'', we mean ``not necessarily total''; so
  total transformations are a special case of partial ones.}

\begin{theorem}[Exact Realization]
\label{thm:exact}\mbox{}
Let $L$ range over nonzero linear operators on $\H$, and let $U$ range
over unitary operators on $\H^+$.
\begin{enumerate}
\item Every $L$ is exactly realizable by some $U$, and one ancilla
  suffices for the purpose.
\item Let $\lambda_{\min}$ and $\lambda_{\max}$ be the minimal and
  maximal eigenvalues respectively of the positive operator $L^\dag
  L$. If $U$ exactly realizes $L$ then the guaranteed success
  probability of $U$ is at most
  $\displaystyle\frac{\lambda_{\min}}{\lambda_{\max}}$, and the upper
  bound is achieved by some unitary operators $U$.
\end{enumerate}
\end{theorem}

Even though Exact Realization Theorem is well-known to experts, we
have found in the literature only a quick proof of Claim~(1), namely
the proof of Claim~6.2 in \cite{AAEL}. For the reader's convenience
and to make this paper more self-contained, we give a detailed proof
of the Exact Realization Theorem in \S\ref{sec:exact}. Specifically we
need the following corollary of Theorem~\ref{thm:exact}.

\begin{corollary}\label{cor:exact}
A partial transformation $T$ of $\H - \{\vec{0}\}$ is exactly
realizable if and only if there is a linear operator $L$ on $\H$ such
that $L\ketpsi$ is nonzero and collinear with $T\ketpsi$ for every
$\ketpsi\in\dom{T}$.
\end{corollary}

Our main concern in this paper is with approximate realizability in
the simple case of our experiment where no restrictions are placed on
the unitary operator $U$ on $\H^+$.
It will be convenient to identify
nonzero collinear $\H$ vectors and work in the resulting complex
projective space $\P$ where each point represents a unique state of
the principal system, and each state is represented by a unique point
in $\P$.
We presume that $\H = \C^n$, so that $\P$ is the complex
projective space of (complex) dimension $n-1$.

We show that, while almost every partial transformation of $\P$ with
domain of cardinality $\le n+1$ is approximately realizable, almost no
partial transformation of $\P$ with a larger domain is approximately
realizable. To formulate this resul t precisely, we need a couple of definitions.

The point in $\P$ given by a nonzero vector $\v$ in $\H$ will be denoted
$\Q\v$. Any linear transformation $L$ of $\H$ induces a partial
transformation $\Q\v\mapsto\Q(L\v)$ of $\P$, denoted $\QL$, with
$\dom{\QL} = \{\Q\v: L\v\ne\vec0\}$. The $\Q$ notation alludes to the
fact that $\P$ is a quotient of $\H -
\{\vec{0}\}$. Corollary~\ref{cor:exact} justifies the following
definition.

\begin{definition}[Exactly realizable transformations of $\P$]
\label{def:exact}
A partial transformation $\tau$ of $\P$ is \emph{exactly realizable}
if there is a linear transformation $L$ of $\H$ such that $\QL$
coincides with $\tau$ on $\dom{\tau}$.
\end{definition}

The complex projective space $\P$ is a Riemannian manifold endowed
with the Fubini-Study metric, the only (up to a nonzero constant
factor) Riemannian metric on $\P$ invariant under (the transformations
of $\P$ induced by) unitary transformations of the overlying Hilbert
space $\H$. The Fubini-Study metric induces the standard Fubini-Study
distance measure
\[
  \FS(\Q\u,\Q\v) = \arccos \frac
  {\lvert \langle \u | \v \rangle\rvert}
  {|\u| \cdot |\v|}.
\]

It will be convenient to represent finite transformations of $\P$ as
point sequences. Fix a positive integer $\ell$. A \emph{suite}
$\sigma$ is a list $(p_1, \dots p_\ell, p_{\ell+1}, \dots, p_{2\ell})$
of $2\ell$ points where the first $\ell$ points $p_1, \dots, p_\ell$
are all distinct; it is a point in the direct product $\P^{2\ell}$ of
the complex projective space $\P$. We think of it as specifying a
transformation
\[
  p_i\mapsto p_{\ell+i}\quad\ \text{where} i=1,\dots,\ell
\]
and we say that it is exactly realizable if the transformation is.
We carry over to suites the standard notation for domain and range of
partial transformations; thus, we write $\dom{\sigma}$ for the first
half, $(p_1,\dots,p_\ell)$, of $\sigma$ and $\range{\sigma}$ for the
second half, $(p_{\ell+1},\dots,p_{2\ell})$.  Notice, though, that a
suite contains more information than just the transformation that it
specifies, because a suite also gives an ordering of its domain.

In what follows, $\eps$ ranges over positive real numbers.

\begin{definition}[Approximately realizable suites]\label{def:appr}
\mbox{}\\[-1em]
\begin{itemize}
\item A suite $(q_1,\dots,q_{2\ell})$ \emph{$\eps$-approximates} a
  suite $(p_1,\dots,p_{2\ell})$ if every $\FS(p_j,q_j) < \eps$.
\item A suite $\sigma$ is \emph{$\eps$-approximable} if there is an
  exactly realizable suite $\tau$ that $\eps$-approximates $\sigma$.
\item A suite is \emph{infinitely approximable} if it is
  $\eps$-approximable for every $\eps>0$.
\end{itemize}
\end{definition}

\begin{theorem}[Approximate Realization]
\label{thm:appr}
In $\P^{2\ell}$, we have the following.
\begin{enumerate}
\item If $\ell\le n+1$ then the set of exactly realizable suites is a set of full measure.
\item If $\ell > n+1$ then $\eps$-approximable suites form an open set of volume $O(\eps^{2(\ell-n-1)(n-1)})$.
\end{enumerate}
\end{theorem}

The proof of Claim~(1) of Theorem~\ref{thm:appr} is elementary but the
proof of Claim~(2) involves the volume-of-the-tube theory pioneered by
Hermann Weyl \cite{Gray,Weyl} and Tarski's theorem about quantifier
elimination in the first-order theory of algebraically closed fields
\cite{Tarski}.

\begin{theorem}[Infinite Approximability]
\label{thm:inf}
  If $\ell < 3$ then every suite is exactly realizable. Assume that
  $\ell \ge3$ but the principal quantum system consists of just one
  qubit. A suite $(p_1,\dots,p_{2\ell})$ is infinitely approximable if
  and only if it is exactly realizable or else exactly $\ell-1$  of the
  $\ell$ points $(p_{\ell+1}, \dots, p_{2\ell})$ are equal.
\end{theorem}

The problem of characterization of infinitely approximable suites in
the general case is open.

\subsection*{Acknowledgments}

We thank Ralf Spatzier for finding useful volume-of-the-tube
references; Nathan Wiebe for his comments on the earlier version of
this paper; and Dorit Aharonov, Vadym Kliuchnikov and Matthew Hastings
for useful discussions.

\section{Prescribing inner products} %2
\label{sec:prelim}

For the reader's convenience, we prove here some well-known facts
about existence of vectors with prescribed inner products. We'll work
over the complex field $\C$.  Except when the contrary is explicitly
stated, vector spaces of the form $\C^d$ are assumed to be equipped
with the standard (for physicists) inner product
\[
\sq{\vec a,\vec b} = \sum_{i=1}^d \overline a_ib_i.
\]

\begin{proposition}\label{pro:inner}
  Let $Q$ be an $n\times n$ matrix of complex numbers.  The following
  statements are equivalent.
  \begin{enumerate}
  \item For some positive integer $d$, there are $n$ vectors $\vec
    x_i\in\C^d$ (where $i=1,2,\dots,n$) such that $\sq{\vec
      x_i,\vec x_j}=Q_{ij}$ for all $i$ and $j$.
  \item There are $n$ vectors $\vec x_i\in\C^n$ (where
    $i=1,2,\dots,n$) such that $\sq{\vec x_i,\vec x_j}=Q_{ij}$ for all
    $i$ and $j$.
  \item $Q$ is Hermitian, and
    $\sum_{i=1}^n\sum_{j=1}^n Q_{ij}\overline z_iz_j\geq 0$ for all
    $\vec z\in\C^n$.
\item $Q$ is Hermitian, and all its eigenvalues are non-negative.
  \end{enumerate}
\end{proposition}

\begin{proof}
We'll prove (1)$\to$(3)$\to$(2), and (3)$\to$(4)$\to$(3).
Since (2) trivially implies (1), this will complete the proof.

\smallskip

 (1)$\to$(3):\quad
Given vectors $\vec x_i$ as in (1), we have
\[
Q_{ij}= \sq{\vec x_i,\vec x_j} =
        \overline {\sq{\vec x_j,\vec x_i}} = \overline{Q_{ji}},
\]
so $Q$ is Hermitian,  We also have, for all $\vec z\in\C^n$, that
\[
\sum_{i,j}Q_{ij}\overline z_iz_j=
\sum_{i,j}\overline z_i\sq{\vec x_i,\vec x_j}z_j=
\sq{\sum_iz_i\vec x_i,\sum_jz_j\vec x_j}\geq0,
\]
where the last inequality comes from the fact that the inner product
of any vector with itself is non-negative.

\smallskip

(3)$\to$(2):\quad Assume (3) and consider an $n$-dimensional
vector space $V$ over $\C$ with a basis $\{\vec e_1,\dots,\vec
e_n\}$.  (To avoid confusion, it is best \emph{not} to identify $V$
with $\C^n$ at this stage; in particular, we do not want the
standard inner product on $V$.)  Define a sesquilinear form (i.e.,
linear in the second argument and conjugate-linear in the first) $B$
on $V$ by setting $B(\vec e_i,\vec e_j)=Q_{ij}$ and extending $B$ to
all vectors in $V$ by sesquilinearity. Because $Q$ is Hermitian, $B$
is conjugate-symmetric, i.e., $B(\vec u,\vec v)=\overline{B(\vec
  v,\vec u)}$.

Observe that the expression $\sum_{i,j}Q_{ij}\overline z_iz_j$, which
we know to be non-negative by (3), is exactly $B(\sum_iz_i\vec
e_i,\sum_i z_i\vec e_i)$.

Temporarily assume that this expression is not only non-negative but
strictly positive for all $\vec z\neq\vec 0$.  Then $B$ is an inner
product on $V$.  So we have an $n$-dimensional complex inner product
space (namely $V$ with inner product $B$) containing $n$ vectors
(namely the $\vec e_i$'s) whose inner products are given by $Q_{ij}$.  But
all $n$-dimensional inner-product spaces over $\C$ are isomorphic,
so the standard such space, $\C^n$ with the standard inner
product, must also contain such vectors.  Thus, we have (2).

It remains to handle the case where $\sum_{i,j}Q_{ij}\overline
z_iz_j$, though non-negative for all $\vec z$ as required in (3),
vanishes for some non-zero vectors $\vec z$.  So $B$ fails to be an
inner product on $V$; it satisfies all the requirements in the
definition of inner products except that
\[
K=\{\vec u\in V:B(\vec u,\vec u)=0\}
\]
is not merely $\{\vec 0\}$.

We claim that $B(\vec u,\vec v)=0$ whenever $\vec u\in K$, for all
$\vec v\in V$.  Indeed, for any such $\vec u$ and $\vec v$ and for any
$\alpha\in\C$, we have
\[
0\leq B(\vec v+\alpha\vec u,\vec v+\alpha\vec u)=
B(\vec v,\vec v)+2\text{Re}(\overline\alpha B(\vec u,\vec v)).
\]
If $B(\vec u,\vec v)$ were not zero, then an appropriate choice of
$\alpha$ would make $\text{Re}(\overline\alpha B(\vec u,\vec v))$ so
negative as to violate this inequality.  This completes the proof of
the claim that $B(\vec u,\vec v)=0$ whenever $\vec u\in K$, for all
$\vec v\in V$.

This claim has two consequences.  First, it tells us that
\[
K=\{\vec u\in V:(\forall\vec v\in V)\,B(\vec u,\vec v)=0\}
\]
and so $K$ is a vector subspace of $V$.  So we can form the quotient
space $V/K$; it is a complex vector space of dimension $<n$.

Second, we have, for arbitrary $\vec u,\vec u'\in K$ and arbitrary
$\vec v,\vec w\in V$, that
\[
B(\vec v+\vec u,\vec w+\vec u')=B(\vec v,\vec w).
\]
This means that $B$ determines a well-defined, conjugate-symmetric,
sesquilinear form $\hat B$ on $V/K$.  That is, if we write $[\vec v]$
for the coset in $V/K$ that contains the vector $\vec v$, then
\[
\hat B([\vec v],[\vec w])=B(\vec v,\vec w)
\]
is well-defined and satisfies all the requirements for an inner
product except perhaps positivity.  It satisfies $\hat B([\vec
v],[\vec w])\geq0$ because of the analogous fact about $B$.  But also,
by dividing out $K$, we have eliminated the danger of equality here.
That is, if $\hat B([\vec v],[\vec v])=0$, then $B(\vec v,\vec v)=0$,
which means $\vec v\in K$ and so $[\vec v]=[\vec 0]$.  So $\hat B$ is
an inner product on $V/K$.

Again, we have a complex inner product space (namely $V/K$ with $\hat
B$) containing $n$ vectors (namely the $[\vec e_i]$'s) whose inner
products are given by the entries of $Q$.  The same therefore holds of
any other complex inner product space of the same dimension, since all
such spaces are isomorphic.  Since $V/K$ has dimension $<n$, we can
find appropriate vectors in $\C^n$ (with room to spare), verifying
(2).

\smallskip

(3)$\to$(4):\quad Since $Q$ is Hermitian, all its eigenvalues are
real.  If one of them were negative, say $\lambda<0$ with eigenvector
$\vec z\neq\vec0$, then
\[
\sum_{i=1}^n\sum_{j=1}^nQ_{ij}\overline z_i z_j=
\sum_{i=1}^n\overline z_i(Q\vec z)_i=
\lambda\sum_{i=1}^n\overline z_iz_i<0,
\]
contradicting the assumption (3).

\smallskip

(4)$\to$(3):\quad
Since $Q$ is Hermitian, there is a unitary matrix $U$ such that
$UQU^\dag$ is a diagonal matrix $D$, whose diagonal entries are the
eigenvalues of $Q$, known to be non-negative by (4).  So we have
$Q=U^\dag DU$.  For any $\vec z\in\C^n$, view $\vec z$ as a column
vector and observe that
\[
\sum_{i,j}Q_{ij}\overline z_i z_j=\vec z^\dag Q\vec z=
\vec z^\dag U^\dag DU\vec z=\vec w^\dag D\vec w,
\]
where we've introduced the abbreviation $\vec w$ for $U\vec z$.  Since
$D$ is diagonal, we have
\[
\vec w^\dag D\vec w=\sum_iD_{ii}\overline w_iw_i,
\]
in which every summand is non-negative.  This completes the
verification of (3) and thus the proof of the proposition.
\end{proof}

\section{Exact Realization Theorem} %3
\label{sec:exact}

We use same name for a linear operator and its matrix when the vector basis is clear from the context.
Let $L$ range over nonzero linear operators on the Hilbert space $\H =
\C^n$ for the principal system.  $L$ is \emph{weakly contracting} if
$\Norm{L\v} \le \Norm{\v}$ for every vector $\v\in\H$. Further, let
$\lambda_{\min}$ and $\lambda_{\max}$ be the minimal and maximal
eigenvalues of the positive operator $L^\dag L$.

\subsection{Literal realization} %3.1
\label{sub:literal}

We start by introducing a particulary simple version of exact
realization. Recall that, according to \S\ref{sec:intro}, every
unitary operator $U$ on the Hilbert space $\H^+$ for the composite
system gives rise to a linear operator $\U$ on $\H$.

\begin{definition}\label{def:literal}
  A unitary operator $U$ on $\H^+$ \emph{literally realizes} $L$ if $L=\U$.
\end{definition}

\begin{proposition}\label{pro:literal}
 The following statements are equivalent.
 \begin{enumerate}
\item $L$ is literally realizable.
\item $L$ is literally realizable with one ancilla.
\item All eigenvalues of $L^\dag L$ are $\le1$.
\item $L$ is weakly contracting.
\end{enumerate}
\end{proposition}

\begin{proof}
Clearly (2)$\to$(1). Taking into account that the Hermitian operator
$L^\dag L$ is diagonalizable, we see that (3) is equivalent to
 \begin{itemize}
 \item[$(3')$] All eigenvalues of $I - L^\dag L$ are $\ge0$.
 \end{itemize}
In the rest of the proof, we establish
$(1)\to(3')\to(2)$ as well as $(3)\leftrightarrow(4)$

Let $k$ be the dimension of the Hilbert space for the auxiliary
system. We work in some basis $\ket0, \dots, \ket{kn-1}$ of $\H^+$. To
simplify the exposition, we presume (without loss of generality
really) that the initial state $\ket{\alpha_1}$ and the final state
$\ket{\alpha_2}$ of the auxiliary system coincide, and that the first
$n$ basic states $\ket0, \dots, \ket{n-1}$ of the composite system are
exactly the basic states where the auxiliary system is in state
$\ket{\alpha_1}$.
According to \S\ref{sec:intro},
\begin{equation*}
 \U\ketpsi = \iota\pi U(\ket{\alpha_1}\otimes\ketpsi)
\end{equation*}
where $\pi$ is the projection
$\H^+ \to \{\ket{\alpha_1}\}\otimes\H$
and $\iota$ is the isomorphism
$\{\ket{\alpha_1}\}\otimes\H \to \H$.

(1)$\to(3')$\quad
Assume $L=\U$.
The matrix $\pi U$ is obtained from matrix $U$ by leaving the upper
$n$ rows intact and zeroing the other entries; the lower $kn-n$ rows
of $U$ play little role in our proof. Further, only the upper $n$
entries of the vector $\ket{\alpha_1}\otimes\ketpsi$ may be nonzero,
and so the right $kn-n$ columns of matrix $U$ play little role in our
proof. If $M$ is the upper left $n\times n$ minor of $U$ then $M\v =
L\v$ for all vectors $\v\in\H$. Thus matrix $L$ is the upper left
minor of matrix $U$.

Let $X$ be the lower left $(kn-n) \times n$ submatrix of $U$ (the
submatrix right under the minor $L$), and let $\ket{L_1}, \dots,
\ket{L_n}$ and $\ket{X_1}, \dots, \ket{X_n}$ be the columns of $L$ and
$X$ respectively. Since $U$ is unitary, we have
$$
\langle X_i | X_j \rangle=
  \begin{cases} -\langle L_i | L_j\rangle & \text{if $i\ne j$,}
\\[0.5em]
1-\langle L_i | L_j\rangle &\;\text{if $i=j$,}
  \end{cases}
$$

\smallskip\noindent
so that the matrix $X^\dag X = I - L^\dag L$. By the implication
(2)$\to$(4) of Proposition~\ref{pro:inner}, with $I - L^\dag L$
playing the role of $Q$, all eigenvalues of $I - L^\dag L$ are
non-negative.

\smallskip\noindent $(3')\to(2)$\quad
Assume $(3')$. By the implication $(4)\to(1)$ of
Proposition~\ref{pro:inner}, with $I - L^\dag L$ playing the role of
$Q$, there exist $n$-dimensional vectors $\ket{X_1}, \dots, \ket{X_n}$
such that the inner products $\langle X_i | X_j \rangle$ form the
matrix $I - L^\dag L$.

Now we are ready to construct the desired matrix $U$. Put $L$ in the
upper left corner of the matrix. Right under $L$ put the $n \times n$
matrix with columns $\ket{X_1}, \dots, \ket{X_n}$. This gives us the
first $n$ columns of $U$ which form an orthonormal basis $B$ for an
$n$-dimensional subspace of $\H^+$.  Extend the list $B$ with the
standard basis $\ket0, \dots, \ket{2n-1}$ for $\H^+$ and then apply
the Gram-Schmidt algorithm to the resulting list in order to obtain an
orthonormal basis for $\H^+$ extending $B$. This basis provides the
columns of the desired matrix $U$.
Thus claims (1),(2),(3) are equivalent. To finish the proof, it
suffices to establish that (3)$\iff$(4).

\smallskip\noindent $(3)\to(4)$\quad
  Assume (3). Let vectors $\ket{e_i}$ form an
  orthonormal basis of eigenvectors of $L^\dag L$ with eigenvalues
  $\lambda_1,\dots,\lambda_n$ respectively. Then $\ketpsi$ is a linear
  combination $\sum_i\alpha_i \ket{e_i}$.
We have
  \begin{align*}
    \Vert L\ketpsi\Vert^2
&\;=(L\ketpsi)^\dag(L\ketpsi)
 = \bra{\psi} L^\dag L \ketpsi\\
&\;= \sum_{i,j} {\overline{\alpha_j}}\,\alpha_i \bra{e_j} L^\dag L \ket{e_i}
=\sum_{i,j}{\overline{\alpha_j}}\,\alpha_i \lambda_i \langle e_j | e_i \rangle\\
&\;=\sum_i |\alpha_i|^2\lambda_i \leq\sum_i|\alpha_i|^2=\Norm{\ketpsi}^2.
  \end{align*}

Next assume (4) and let $\ketpsi$ be an eigenvector of $L^\dag L$ with
some eigenvalue $\lambda$.  Then
\begin{align*}
\lambda\Vert \ketpsi\Vert^2
&\; = \bra{\psi}\lambda \ketpsi
 = \bra{\psi} L^\dag L \ketpsi\\
&\; = (L\ketpsi)^\dag(L\ketpsi)
 = \Vert L\ketpsi\Vert^2
 \leq \Vert \ketpsi \Vert^2,
\end{align*}
so $\lambda\leq1$.
\end{proof}

\begin{comment}
\begin{corollary}\label{cor:mult1}
  A given operator $\gamma L$ is literally realizable if and only if
  $|\gamma|^2\lambda_{\max}\le1$ for some nonzero complex number
  $\gamma$.
\end{corollary}

\begin{proof}
  Let $F = \gamma L$. Then $F^\dag F = L^\dag \overline{\gamma} \gamma
  L = |\gamma|^2 L^\dag L$, and the maximal eigenvalue of $F^\dag F =
  |\gamma|^2\lambda_{\max}$. Now use
  Proposition~\ref{pro:literal}.
\end{proof}
\end{comment}

\begin{corollary}\label{cor:convex}\cite{CW}
Any convex combination of unitary operators is literally realizable.
\end{corollary}

\begin{proof}
  Suppose $L$ is a convex combination of some unitary operators
  $U_i$. Then, for any vector $\ketpsi$, $L\ketpsi$ is a convex
  combination (with the same coefficients) of the vectors
  $U_i\ketpsi$, each of which has the same length as $\ketpsi$.
  Since balls in Hilbert space are convex, it follows that $L\ketpsi$
  has at most the same length as $\ketpsi$.
  So $L$ is weakly contracting. By Proposition~\ref{pro:literal}, $L$
  is literally realizable.
\end{proof}

Childs and Wiebe prove more in \cite{CW}. In particular, the $U$ that
literally realizes a convex combination of two $U_i$'s can be computed
by a circuit consisting of (a)~unitary operators that act only on the
ancilla and (b)~the controlled $U_i$ gates.

\subsection{Literal realization vs. exact realization} %3

\begin{proposition}\label{pro:exact2literal}
  A unitary operator $U$ on $\H^+$ exactly realizes a nonzero linear
  operator $L$ if and only if it literally realizes some nonzero
  multiple $cL$ of $L$.
\end{proposition}

\begin{proof}
  The if part of the proposition is obvious: if $cL = \U$ then $L\v$
  and $\U\v$ are collinear for every $\v\in\H$. To prove the only-if
  part, we need an auxiliary result from linear algebra.

\begin{lemma}\label{lem:collinear}
Let $D, R$ be finite-dimensional complex vector spaces, and let $A, B$
be linear transformations from $D$ to $R$ such that $A\v$ and $B\v$
are collinear for every $\v\in D$. Then $A,B$ are collinear, that is
$A = cB$ for some nonzero $c$.
\end{lemma}

\begin{proof}[Proof of Lemma~\ref{lem:collinear}]
  First we treat the case where $B$ is one-to-one. Let $d$ be the
  dimension of $D$. If $d=1$, the lemma is obvious, so we may assume
  that $d\ge2$. Let vectors $\vec v_1, \dots, \vec v_d$ in $D$ form a
  basis in $D$. Since $B$ is one-to-one, the vectors $B\vec v_i$ are
  linearly independent. By the collinearity premise, there are nonzero
  complex numbers $c_i$ such that $A\vec v_i = c_i B\vec v_i$. It
  suffices to show that all the numbers $c_i$ are equal.

For any $i < j$, let $\vec u = \vec v_i + \vec v_j$. By the
collinearity premise, $A\vec u = cB\vec u$ for some $c$.  We
have \begin{align*} A\vec u = &\; cB\vec u = cB\vec v_i + cB\vec
  v_j,\\ A\vec u = &\; A(\vec v_i + \vec v_j) = c_iB\vec v_i +
  c_jB\vec v_j, \end{align*} so that $cB\vec v_i + cB\vec v_j =
c_iB\vec v_i + c_jB\vec v_j$.  But vectors $B\vec v_i, B\vec v_j$ are
independent. Then $c = c_i$ and $c = c_j$ and therefore $c_i = c_j$.

Second we treat the case where $B$ is not one-to-one.
Without loss of generality we
may suppose that $B$ is nonzero. Clearly, $A\vec u = \vec 0$ whenever
$B\vec u=\vec0$.  That is, $A$ vanishes on the kernel $K$ of $B$. So
we can regard both $A$ and $B$ as being defined on the quotient
$D/K$, and of course $B$ is one-to-one on $D/K$, so that the preceding
discussion applies.

Therefore there is a nonzero complex number $c$ such that $A = cB$ on
$D/K$. We check that $A = cB$ on $D$. Pick any nonzero vector
$\vec v\in D$. Obviously $A\vec v = cB\vec v$ if $\vec v\in
K$. Suppose that $\vec v\notin K$. By the collinearity premise, $A\vec
v = c'B\vec v$ for some nonzero complex number $c'$. This equality
results in a similar equality $A[\vec v] = c'B[\vec v]$ in the
quotient $D/K$ where we also have $A[\vec v] = cB[\vec v]$. Since
vector $B[\vec v]$ is not zero, it follows that $c' = c$.
\end{proof}

Now we are ready to prove the only-if part of the proposition. Assume
that $U$ exactly realizes $L$, so that $\U\v$ is nonzero and collinear
with $L\v$ whenever $L\v\ne\vec0$. If the implication
\[
  L\v = \vec0 \to \U\v = \vec0.
\]
holds then, by Lemma~\ref{lem:collinear}, some nonzero multiple $cL$
of $L$ coincides with $\U$ and therefore is literally
realizable. Thus, it suffices to prove the implication.

Suppose $L\v =\vec0$. Since $L$ is nonzero, there is a vector $\w$
orthogonal to the kernel of $L$, so $L\w$ and $L(\v + \w)$ are equal
and nonzero. Hence $\U\w$ and $\U(\v + \w)$ are nonzero and collinear,
so $\U\v = b\U\w$ and therefore $\U(\v-b\w) =\vec0$ for some $b$. But
then $L(\v-b\w) = \vec0$, so that $b=0$ and $\U\v=\vec0$.
\end{proof}

\subsection{Guaranteed success probability} %3.3
\label{sub:gsp}

\begin{proposition}\label{pro:gsp}
  If a unitary operator $U$ literally realizes $L$ then the guaranteed
  success probability of $U$ is the least eigenvalue
  $\lambda_{\min}$ of $L^\dag L$.
\end{proposition}

\begin{proof}
  Let $ \ketpsi$ range over the unit sphere of $\C^n$. Recall from
  \S\ref{sec:intro} that the guaranteed success probability of $U$ is
  $\min_{\ketpsi}\SP(U,\ketpsi)$ where $\SP(U,\ketpsi) = \Norm{\U\ketpsi}^2$,
and assume that $U$ literally realizes $L$.
Then
$$
  \SP(U,\ketpsi) = \Norm{L\ketpsi}^2.
$$

There exist eigenvectors $\ket{e_1}, \dots, \ket{e_n}$ of $L^\dag L$,
with eigenvalues $\lambda_{\max} = \lambda_1 \ge \cdots \ge \lambda_n
= \lambda_{min}$ respectively, that form an orthonormal basis for
$\H$. An arbitrary unit vector $\ketpsi$ in $\H$ is a linear
combination $\sum_{i=1}^n\alpha_i \ket{e_i}$. We have
  \begin{align*}
    \SP(U,\ketpsi)
&\;= \Vert L\ketpsi\Vert^2
 =(L\ketpsi)^\dag(L\ketpsi)
 = \bra{\psi} L^\dag L \ketpsi\\
&\;= \sum_{i,j}\overline{\alpha_j}\alpha_i \bra{e_j} L^\dag L \ket{e_i}
 = \sum_{i,j} \overline{\alpha_j}\alpha_i \lambda_i \langle e_j | e_i \rangle\\
&\;=\sum_i |\alpha_i|^2\lambda_i
    \ge \lambda_{\min}\sum_i|\alpha_i|^2 = \lambda_{\min}.
  \end{align*}
  In particular $\SP(U,\ket{e_n}) = \lambda_n = \lambda_{\min}$.
\end{proof}

\begin{corollary}
  Suppose that $U$ literally realizes $L$. Then $L$ is invertible if and only if the guaranteed success probability of $U$ is positive.
\end{corollary}

%This follows from Propositions~\ref{pro:exact2literal} and \ref{pro:gsp}.

\subsection{Proof of Exact Realization Theorem} %3.4
\label{sub:exact}

\begin{proof}[Proof of Claim~(1) of Theorem~\ref{thm:exact}]
If some nonzero multiple $cL$ of the given linear operator $L$ on $\H$
is literally realizable then, by Proposition~\ref{pro:literal}, $cL$
is literally realizable by some unitary operator $U$ with just one
ancilla. But then $U$ exactly realizes $L$, and one ancilla
suffices. Thus it suffices to find a complex number $c\ne0$ such that
$cL$ is literally realizable.

If $\lambda_{\max}\le1$ set $c=1$; otherwise set $c =1
/\sqrt{\lambda_{\max}}$. In either case, by
Proposition~\ref{pro:literal}, $cL$ is literally realizable.
\end{proof}

\begin{proof}[Proof of Claim~(2) of Theorem~\ref{thm:exact}]
Assume that $U$ exactly realizes $L$. By
Proposition~\ref{pro:exact2literal}, $U$ literally realizes some
nonzero multiple $M = cL$ of $L$. Let $\mu_{\min}$ and $\mu_{max}$ be
the minimal and maximal eigenvalues of $M^\dag M$ respectively. Taking
into account that $M$ is nonzero and invoking
Proposition~\ref{pro:literal}, we have
\[
  0 < \mu_{\max} = |c|^2 \lambda_{\max} \le 1.
\]
By Proposition~\ref{pro:gsp}, the guaranteed success probability of $U$ is \[
  \mu_{\min} = |c|^2 \lambda_{\min} \le \lambda_{\min}/\lambda_{\max}.
\]
There is a real $d\ge1$ such that $|cd|^2 \lambda_{\max} =
1$. Redefine $M$ from $cL$ to $cdL$. The unitary $dU$ literally
realizes $M$ and therefore exactly realizes $L$. We have
\[
  0 < \mu_{\max} = |cd|^2 \lambda_{\max} = 1,
\]
and the guaranteed success probability of $dU$ is
\[
  \mu_{\min} = |cd|^2 \lambda_{\min} = \lambda_{\min}/\lambda_{\max}.
\]
\end{proof}

\section{Approximate Realization Theorem} %4
\label{sec:appr}

In this section and in the rest of the paper, we use notation and
definitions from \S\ref{sec:intro}. In particular, every nonzero
vector $\v = (a_1,\dots,a_n)$ in $\H$ represents a point $\Q\v$ in
$\P$. The complex numbers $a_1,\dots,a_n$ are the homogeneous
coordinates of $\Q\v$; at least one of the homogeneous coordinates is
nonzero. Further, any linear transformation $L$ of $\H$ induces a
partial transformation $\QL$ of $\P$. If $L$ is invertible then $\QL$
is total. Such total transformations $\QL$ are known as projective
linear.

\subsection{Projective linear transformations} %4.1
\label{sub:pl-trans}

As usual, nonzero vectors $\v_1, \dots, \v_m$ in $\C^n$ are said to be
\emph{in general position} if, for any $k\le n$, any $k$ of the $m$
vectors are linearly independent. Points $\Q\v_1, \dots, \Q\v_k$ are
\emph{in general position} if the vectors $\v_1, \dots, \v_k$ are so.

\begin{lemma}\label{lem:n+1}
  If points $p_1,\dots,p_{n+1}$ are in general position and points
  $q_1,\dots,q_{n+1}$ are in general position then there is a unique
  projective linear transformation $g$ such that every $g(p_i) =
  q_i$.
\end{lemma}

\begin{proof}
Let $\e_{n+1}$ be the sum $\e_1 + \cdots \e_n$ of the basic vectors in
$\H$. It is easy to check that vectors $\e_1,\dots,\e_{n+1}$ are in
general position. It suffices to prove that for any vectors
$\v_1,\dots,\v_{n+1}$ in general position there exists a unique, up to a
constant factor, invertible linear operator $L$ on $\H$ such that
every $L\e_i$ is collinear with $\v_i$.

First we prove the uniqueness.  Suppose that $L$ is a linear operator
such that every $L\e_i$ is collinear with $\v_i$, and so there are
nonzero complex numbers $z_i$ such that
\begin{equation*}
L\e_i = z_i \v_i \text{\quad for }i=1,\dots,n+1.
\end{equation*}
In the basis $\e_1,\dots,\e_n$, the column vector $\e_i$ with $i\leq n$
has 1 at row $i$ and zeroes everywhere else, so that the $i\th$ column
of the desired $L$ is $z_i\v_i$. Since $\e_{n+1} = \e_1 + \cdots + \e_n$,
we have
  \begin{equation*}
    z_{n+1}\v_{n+1} = \sum_{i=1}^n z_i\v_i \,.
  \end{equation*}
Since vectors $\v_1,\dots,\v_n$ are independent, $\v_{n+1} = a_1\e_1
+\cdots+ a_n\e_n$ for some complex numbers $a_1, \dots, a_n$, so that
$z_1 = a_1z_{n+1}, \dots, z_n = a_nz_{n+1}$.
Since vectors $\v_0, \dots, \v_n$ are in general position, the
coefficients $a_1, \dots, a_n$ are nonzero.
Let $L_0$ be the invertible matrix with columns
$a_1\v_1,\dots,a_n\v_n$. Then $L = z_{n+1} L_0$.

Second we prove the existence. To this end, check that every $L_0\e_1
= a_1\v_1, \dots, L_0\e_n = a_n\v_n$ and $L_0\e_{n+1} = \v_{n+1}$.
\end{proof}

Recall that suites are points of $\P^{2\ell}$ where the first $\ell$
coordinates are distinct and that a suite $\tau = (p_1, \dots,
p_{2\ell})$ specifies the transformation $\tau(p_j) = p_{\ell+j}$ with
domain $\{p_1, \dots, p_\ell\}$.

\begin{definition}[PL manifold]
The \emph{projective linear manifold} PL consists of the suites
specifying partial transformation of $\P$ that can be extended to
projective linear transformations of $\P$.
\end{definition}

We say that, over $\H$, a sequence $M_1, M_2, \dots$ of linear
operators converges to a linear operator $L$ if, for every vector
$\v$, the sequence $M_i\v$ converges to $L\v$.

\begin{lemma}\label{lem:approximants}
  Over $\H$, for every linear operator $L$ on $\H$ there is a
  sequence $M_1, M_2, \dots$ of invertible linear operators that
  converges to $L$.
\end{lemma}

\begin{proof}
  Without loss of generality, $L$ is positive. Indeed, by the Polar
  Decomposition Theorem, $L = UL'$ for some unitary $U$ and positive
  $L'$. If invertible linear operators $M'_i$ converge to $L'$ then
  $UM'_i \to UL'$.

  Fix an orthonormal basis for $\H$ composed of eigenvectors of
  $L$. In that basis, $L$ is represented by a diagonal matrix. The
  (matrix for the) desired $M_i$ is obtained from $L$ by replacing
  every zero on the diagonal with $1/i$.
\end{proof}

\begin{proposition}[PL approximants suffice]
\label{pro:approximants}
For every $\eps$-approximable suite $\sigma$ there is a PL suite that
$\eps$-approximates $\sigma$.
\end{proposition}

\begin{proof}
Given an $\eps$-approximable suite $\sigma$, first choose an exactly
realizable suite $\tau$ that $\eps$-approximates $\sigma$.  Let
$\delta$ be the maximum of the Fubini-Study distances between
corresponding components of $\sigma$ and $\tau$.  So $\delta<\eps$.
By Corollary~\ref{cor:exact}, we have a linear operator $L$ that
realizes $\tau$.  By Lemma~\ref{lem:approximants}, we can find
invertible linear operators $M$ arbitrarily close to $L$.  Taking $M$
close enough to $L$, we can ensure, thanks to the continuity of the
quotient map $Q:\H\to\P$, that $QM$ maps each point in $\dom \tau$ to
within $\eps-\delta$ of the corresponding point in $\range \tau$.
Then, letting $\tau'$ be the suite with the same domain half as
$\tau$ but the range half given by applying $QM$ to the domain, we get
that $\tau'$ is within $\eps-\delta$ of $\tau$ and therefore within
$\eps$ of $\sigma$.
\end{proof}

\subsection{The PL manifold} %4.2
\label{sub:manifold}

The complex projective space $\P$ has dimension $n-1$. So
$\dim(\P^{2\ell}) = 2\ell(n-1)$.

\begin{lemma}[Dimension of PL]\label{lem:dim}\mbox{}\\[-1em]
\begin{enumerate}
\item If $\ell \le n+1$ then PL is an open set of full measure in
  $\P^{2\ell}$, and so $\dim(\PL) = 2\ell(n-1)$.
\item If $\ell > n+1$ then $\dim(\PL) \le (n-1)(\ell+n+1)$.
\end{enumerate}
\end{lemma}

\begin{proof}[Proof of Lemma]
Claim~(1) follows from Lemma~\ref{lem:n+1}. We prove Claim~(2).

A PL suite $\tau = (p_1, \dots, p_{2\ell})$ is determined by
$p_1,\dots,p_\ell$ and an invertible linear operator $L$ on $\H =
\C^n$ such that $\QL(p_j) = p_{\ell+j}$ for $j\le\ell$.  So PL is the
range of a (smooth, in fact rational in local coordinates) map from
$\P^\ell \times \L$ where $\L$ is the space of linear operators on
$\C^n$ modulo scalar multiples. Thus
\[
\dim(\PL)\leq\dim(\P^\ell) + \dim(\L) =
             \ell(n-1) + (n^2-1) = (n-1)(\ell+n+1).
\]
This completes the proof of the lemma.
\end{proof}

We remark that the upper bound in Claim~(2) of the lemma is, in all
nontrivial cases (i.e., $n>1$), strictly below the dimension of
$\P^{2\ell}$.

It will be convenient to work in affine spaces rather than projective
ones.  To this end, cover the complex projective space $\P$ by its $n$
standard coordinate patches $A_1, \dots, A_n$. Here $A_i$ consists of
those points in $\P$ whose $i\th$ homogeneous coordinate is not zero;
multiplying by a scalar, we can arrange that the $i\th$ homogeneous
coordinate is 1, and then we can use the remaining $n-1$ homogeneous
coordinates as affine coordinates on $A_i$.  The space $\P^{2\ell}$ is
covered by $n^{2\ell}$ coordinate patches $A$ that are the cartesian
products of the coordinate patches in the $2\ell$ factors.

\begin{proposition}[Variety for PL]\label{pro:variety}
In any coordinate patch $A$ of $\P^{2\ell}$, there exist a
full-measure open set G and an algebraic variety $V$ such that $\PL
\subseteq V$ and $G \cap V \subseteq \PL$.
\end{proposition}

Here an algebraic variety is the set of solutions of a system of
polynomial equations over the field $\C$ of complex numbers. Notice
that the union of two varieties is a variety. For example,
\begin{align*}
  & (f_1=0 \land f_2=0) \lor (g_1=0 \land g_2=0) \iff \\
  & (f_1g_1=0) \land (f_1g_2=0) \land (f_2g_1=0) \land (f_2g_2=0).
\end{align*}

\begin{proof}
It is clear from the definition of PL that the intersection $A\cap\PL$
is definable, in terms of the affine coordinates of $A$, in the
first-order language of the field $\C$. By Tarski's theorem
\cite{Tarski}, the first-order definition of $A\cap\PL$ can be
rewritten in quantifier-free form.  We can also arrange that the
quantifier-free definition is in disjunctive normal form, and we can
assume that each disjunct is satisfied by some points, because any
other disjuncts could simply be omitted from the disjunctive normal
form.

Any disjunct $\delta$ is a conjunction of some polynomial equations
and some inequations. (``Inequation'' here means $\neq$, whereas
``inequality'' traditionally means $<,>,\leq,\geq$.)  We can further
arrange that there is at most one inequation among the conjuncts,
because $(f\neq0)\land(g\neq0)$ is equivalent to $fg\neq0$, and that
there is at least one inequation, because if there is none then we can
adjoin $1\neq0$.  So $\delta$ has the form
\[
  f_1 = f_2 = \cdots = f_k = 0 \land g\neq0,
\]
where all $f_j$ as well as $g$ are polynomials. Let $E_\delta
=\Set(g=0)$. Here $\Set(g=0)$ is the solution set for $g=0$ in $A$; we
will use similar notation for other formulas as well. Further, let $E$
be the union, over all disjuncts $\delta$, of the sets $E_\delta$. The
desired full-measure open set is $G = A-E$.

Since there are only finitely many disjuncts, it suffices to prove
that, for any disjunct $\delta$, there is an algebraic variety $V$ in
$A$ such that $\Set(\delta) \subseteq V$ and $G\cap V\subseteq\PL$ is
a set of measure 0.

To obtain the desired $V$, we simply remove the inequation from
$\delta$, so that $V =\Set(f_1=\cdots=f_k =0)$.   This looks
simplistic but it works. Obviously $V$  is an algebraic variety and
$V$ includes $\Set(\delta)$. Further,
\begin{align*}
 V & = [V\cap\Set(g\ne0)] \cup [V\cap\Set(g=0)]\\
   & = \Set(\delta) \cup [V\cap\Set(g=0)]\\
   & \subseteq \PL \cup E_\delta
\end{align*}
and therefore $G\cap V\subseteq \PL$.
\end{proof}

\subsection{Tubes} %4.3
\label{sub:tubes}

The purpose of this subsection is to provide some information on tubes
that we'll need in the proof of Approximate Realization Theorem.

Consider a one-dimensional curve $L$ in a three-dimensional cube. Given a small $\eps>0$ and a point $p\in L$, form a disc of radius $\eps$, within the ambient cube, centered at $p$ and orthogonal to $L$ at $p$. (For simplicity, we ignore the possibility that the disc bulges beyond the cube. More pedantically, we should be talking about the portion of the disc within the cube.) As $p$ traverses the curve $L$, the disc traverses a three-dimensional tube of radius $\eps$ around $L$.

Similarly the $\eps$-approximable points of $\P^{2\ell}$ form a tube,
$\Tube_\eps$ around the set of exactly approximable suites. By
Proposition~\ref{pro:approximants}, it is also a tube around the PL
manifold.

If the curve $L$ is nice, one may expect that the volume of the
three-dimensional tube of radius $\eps$ is about $\pi\eps^2$ times the
length of $L$. One can estimate the volume of $\Tube_\eps$ in a
similar way. But the curve $L$ can be so curly that its length is
infinite; the classical example is the curve $\sin(1/x)$ where $0<x<1$
which has a singularity at $x=0$. The curve may even fill in the whole
cube. Of course, the PL manifold does not look bad, and there is a
well developed theory of tubes around a submanifold of a given
manifold \cite{Gray}. The PL manifold has singularities, and it is not
obvious at all how to apply the general theory of tubes. Fortunately,
by Proposition~\ref{pro:variety}, the PL manifold is, up to a set of
measure zero, a finite union of algebraic varieties. That helps.

\begin{theorem}[Wongkew's theorem \cite{W}]\label{thm:W}
Let $V$ be an algebraic variety of codimension $k$ in the
$m$-dimensional Euclidean space given by polynomials of degree $\le
d$. Then there exist constants $c_k, \dots, c_m$ which depend only
on $m$, so that for any ball $B$ of radius $R$ and any positive
$\eps$, the volume of the $\eps$-tube around $B\cap V$ is bounded by
$\displaystyle{\sum_{j=k}^m (c_j d^j R^{m-j}) \eps^j}$.
\end{theorem}

\begin{corollary}\label{cor:W}
Let $V$ be an algebraic variety of complex codimension $k$ in a
finite-dimensional Hilbert space, and let $\eps$ range over positive
reals. For any bounded open set $X$ (or its closure), the volume of
the tube of radius  $\eps$ around $X\cap V$ is $O(\eps^{2k})$.
\end{corollary}

\begin{proof}
We get $O(\eps^{2k})$ by taking into account only the leading term in the sum in Wongkew's theorem. The exponent is doubled because we work in a Hilbert space and $k$ is the complex dimension whereas Wongkew's theorem refers to real dimensions. Finally, balls can be replaced with bounded open sets because the closure of such a set is compact and therefore is covered by finitely many balls.
\end{proof}

\subsection{Proof of Approximate Realization Theorem}
\label{sub:appr}

\begin{proof}[Proof of Claim~(1) of Theorem~\ref{thm:appr}]
Assume $\ell\le n+1$ and let $G$ be the set of general-position suites, so that a suite $(p_1,\dots, p_{2\ell}) \in G$ if and only if the $2\ell$ points $p_1,\dots, p_{2\ell}$ are in general position. Clearly $G$ is an open set of measure 1. By Proposition~\ref{pro:approximants}, $G\subseteq\PL$, so the PL manifold is of full measure. But every PL suite is exactly realizable. So the set of exactly realizable suites is of full measure.
\end{proof}

\begin{proof}[Proof of Claim~(2) of Theorem~\ref{thm:appr}]
Assume $\ell\ge n+2$, and define the \emph{unit cube} of a Hilbert
space $\C^m$ to comprise the points in $\C^m$whose coordinates all
have absolute values in the real interval $[0,1]$.

Recall the coordinate patches of $\P$ and $\P^{2\ell}$ that we used in
the proof of Proposition~\ref{pro:variety}.
If you identify a patch $A_i$ of $\P$ with a copy of $\C^{n-1}$, then
it makes sense to speak about a unit cube $C_i$ in $A_i$. The cubes
$C_1, \dots, C_n$ cover $\P$. Indeed, if $p\in\P$ and the $i\th$
homogeneous coordinate of $p$ is, in absolute value, a largest
homogeneous coordinate of $p$ then $p$ belongs to the cube $C_i$.

Now, every coordinate patch $A$ of $\P^{2\ell}$ is a cartesian product of coordinate patches in the $2\ell$ factors of $\P^{2\ell}$. View every factor patch as a copy of $\C^{n-1}$. Then $A$ is a copy of $\C^{2\ell(n-1)}$. The unit cube $C$ in $A$ is the cartesian product of the unit cubes in the $2\ell$ factor patches. It follows that $\P^{2\ell}$ is covered by the unit cubes in its coordinate patches.

Since the number of such cubes is finite, it suffices to prove the
following for every coordinate patch $A$ of $\P^{2\ell}$: The
$\eps$-approximable suites form an open set of volume
$O(\eps^{2(\ell-n-1)(n-1)})$ in the unit cube $C$ of $A$. By
Proposition~\ref{pro:approximants}, it suffices to prove that, in $C$,
the volume of the $\eps$-tube around $C\cap\PL$ is
$O(\eps^{2(\ell-n-1)(n-1)})$.

To this end, let $G$ and $V_1,\dots,V_k$ be as in
Proposition~\ref{pro:variety}. Then $C\cap G$ is an open set of
measure 1 in the cube $C$ and  $(C\cap G)\cap\PL = (C\cap G) \cap
(V_1\cup\cdots\cup V_k)$.

Since the set $C-G$ is of measure 0, it suffices to prove that, in
$C\cap G$, the volume of the $\eps$-tube around $(C\cap G)\cap\PL$ is
$O(\eps^{2(\ell-n-1)(n-1)})$. To this end, it suffices to prove that,
in $C\cap G$, the volume of the $\eps$-tube around every $C\cap V_j$
is $O(\eps^{2(r-1)(n-1)})$, but this follows directly from
Corollary~\ref{cor:W}.
\end{proof}

\section{Infinite approximability}
\label{sec:inf}

Recall that a suite is infinitely approximable if it is
$\eps$-approximable for every positive $\eps$.  By
Proposition~\ref{pro:approximants}, the infinitely approximable suites
form the closure of PL.

We start with a couple of general remarks and then give a complete
characterization of infinitely approximable suites in the case where
the principal quantum system consists of a single qubit. The problem
of characterization of infinitely approximable suites in the general
case is open.

\subsection{General considerations}
\label{sub:prelim}

Recall that the Hilbert space $\H$ for our principal quantum system is
$\C^n$ and the corresponding projective space is $\P =
\CP^{n-1}$. Suites are tuples in $\P^{2\ell}$ where the first $\ell$
points are distinct. A suite $\sigma = (p_1,\dots,p_{2\ell})$ is
viewed as the finite transformation that sends the domain tuple
$\dom{\sigma} = (p_1,\dots,p_\ell)$ to the range tuple $\range{\sigma}
= (p_{\ell+1}, \dots, p_{2\ell})$.  If $\sigma$ is infinitely
approximable then every vicinity of $\sigma$ contains exactly
realizable suites; but $\sigma$ itself does not have to be exactly
realizable.

\begin{example}[An infinitely approximable suite that is not exactly
  realizable]
Set $\ell = n+1$. In this case, by Lemma~\ref{lem:n+1}, every suite in
general position extends to a projective linear transformation and
thus is exactly realizable. Since general-position suites form an open
set of full measure, every suite is infinitely approximable. It
remains to construct a suite $\sigma = (p_1, \dots, p_{2n+2})$ that is
not exactly realizable.

Construction. Given an orthonormal basis $\e_1,\dots,\e_n$ for $\H$,
let $\e_{n+1} = \sum_{i=1}^n \e_i$. We saw, in the proof of
Lemma~\ref{lem:n+1}, that the vectors $\e_1,\dots, \e_{n+1}$ are in
general position. Set
\begin{align*}
  & p_1 = \Q\e_1, \dots, p_n = \Q\e_n, p_{n+1} = \Q\e_{n+1}\\
  & \sigma(p_1) = \cdots = \sigma(p_n) = p_1, \sigma(p_{n+1}) = p_2.
\end{align*}

By reduction to absurdity, assume that $\sigma$ extends to the
transformation $\QL$ for some linear operator $L$ on $\H$. Then the
vectors $L\e_1, \dots, L\e_n$ are collinear with $\e_1$, and so the
range of $L$ is the one-dimensional subspace spanned by
$\e_1$. Accordingly the range of $\QL$ consists of a single point
$p_1$ while the range of $\sigma$ contains $p_2$ as well.\qed
\end{example}

\begin{lemma}   \label{continuity}
  For each suite $\sigma$ of length $2n+2$, whose domain half
  $\dom{\sigma}$ and range half $\range{\sigma}$ are each in general
  position, let $f_\sigma$ be the unique projective linear
  transformation that maps $\dom{\sigma}$ to $\range{\sigma}$. Then
  $f_\sigma(p)$ is a continuous function of $\sigma$ (in the space of
  suites) and $p$ (in $\P$).
\end{lemma}

\begin{proof}
  Because continuity is a local property, we may assume that the
  relevant points, namely the $2n+2$ components of $\sigma$ and the
  point $p$, are each confined to lie in one of the $n$ coordinate
  patches that cover $\P$.  (Of course, different components might be
  in different patches.)  Fixing these patches, we can fix a
  normalization for the homogeneous coordinates of the relevant
  points.  If a point is confined to the patch where the $i\th$
  homogeneous coordinate is non-zero, then we normalize its
  homogeneous coordinates so that the $i\th$ coordinate is 1.

We now revisit the proof of Lemma~\ref{lem:n+1}, paying attention to
continuity issues.

To begin, consider suites of the form $(E,\range{\sigma})$, where the range
is that of a variable $\sigma$, as above, but the domain is fixed as
the $(n+1)$-tuple $(Q\vec e_1,\dots,Q\vec e_{n+1})$ of points in $\P$
corresponding to the $n$ standard basis vectors $\vec e_1,\dots,\vec
e_n$ and their sum $\vec e_{n+1}$ in $\H$.  Because $\range{\sigma}$ is in
general position, the proof of Lemma~\ref{lem:n+1} produces an
invertible matrix $L$ corresponding to the projective linear
transformation $f_{(E,\range{\sigma})}$, and now we need to look more
closely at this $L$.  (Recall that it is unique up to an overall
nonzero scalar factor.)  It can be obtained as follows.  First form
the matrix $L'$ whose columns are the homogeneous coordinates
(normalized as above) of the first $n$ components of $\range{\sigma}$.  The
corresponding projective linear transformation transforms each $Q\vec
e_i$ for $i=1,2,\dots,n$ correctly, namely to $\range{\sigma}_i$, but it
might transform $Q\vec e_{n+1}$ incorrectly.  To correct this one
remaining component, without damaging the other $n$, we multiply the
columns of $L'$ by suitable nonzero scalars $z_i$.  Any choice of
$z_i$'s will preserve the correctness of the values at $Q\vec e_i$ for
$i=1,2,\dots,n$, but the $z_i$'s must be chosen carefully to ensure
that the new matrix $L$ sends $e_{n+1}$ to $\range{\sigma}_{n+1}$.  (It
would suffice to send $e_{n+1}$ to a vector collinear with
$\range{\sigma}_{n+1}$, but this additional freedom is just the freedom,
already noted above, to multiply $L$ by an overall nonzero scalar
factor.)  The required condition on the $z_i$'s is a system of linear
equations, whose coefficient matrix is $L'$.  The fact that $L'$ is
invertible ensures not only that there is a unique solution for the
$z_i$'s but also that this solution is a continuous function of
$\range{\sigma}$.  Indeed, by Cramer's rule, the solution is given by
certain rational functions, namely ratios of determinants, of the
entries of $L'$ and the components of $\range{\sigma}_{n+1}$.  Since the
entries of $L'$ are components of $\range{\sigma}_i$ for $i=1,\dots,n$, and
since the denominator of these rational expressions, the determinant of
$L'$, is not zero, we have the claimed continuity of the $z_i$'s.  It
follows that $L$ is a continuous function of $\range{\sigma}$.

Similarly, we can realize the finite transformation $(\dom{\sigma},E)$ by
a matrix $M$ whose entries are continuous functions of $\dom{\sigma}$.
Indeed, the previous paragraph shows how to continuously realize
$(E,\dom{\sigma})$.  To realize $(\dom{\sigma},E)$, we need only take the
inverse matrix.  It will still be a continuous function of
$\dom{\sigma}$, because matrix inversion is a continuous function, given
by ratios of determinants.

Having realized both $(E, \range{\sigma})$ and $(\dom{\sigma},E)$ by
matrices that depend continuously on $\sigma$, we need only multiply
these matrices (and observe that multiplication is continuous) to
realize $\sigma$.

Finally, $f_\sigma(p)$ can be obtained as the image in $\P$ of the
product of the matrix realizing $\sigma$ and the column vector
(normalized as above) representing $p$.  It is therefore a continuous
function of $\sigma$ and $p$.
\end{proof}

The \emph{border} of PL consists of the infinitely approximable suites
that do not belong to PL.

\begin{claim} \label{cla:range}
If $\sigma$ is a suite on the border of PL, then there cannot be $n+1$
points in general position in $\dom{\sigma}$ such that the corresponding
$n+1$ points in $\range{\sigma}$ are also in general position.
\end{claim}

\begin{proof}
Suppose that $\sigma=(p_1,\dots,p_l,q_1,\dots,q_l)$ were a
counterexample.  To simplify the notation, permute the components, if
necessary, so that $(p_1,\dots,p_{n+1})$ and $(q_1,\dots,q_{n+1})$ are
general-position $(n+1)$-tuples.  By Lemma~\ref{lem:n+1}, let $f$ be
the unique projective linear transformation that sends $p_i$ to $q_i$
for all $i$ in the range $1\leq i\leq n+1$.  We shall show that
$f(p_j)=q_j$ also for $n+1<j\leq l$.   This will complete the
proof, because it means that $f$ realizes $\sigma$ and therefore
$\sigma$ belongs to PL, not to its border as assumed.  For the rest of
the proof, we fix some arbitrary $j$ in the relevant range, $n+1<j\leq
l$, and our goal is to prove that $f(p_j)=q_j$.

Since $\sigma$ is in the closure of PL, we can consider PL suites
$\sigma'$ arbitrarily close to $\sigma$.  Temporarily consider a fixed
$\sigma'$ near $\sigma$.  (Later, we shall let $\sigma'$ vary and
approach $\sigma$.)  Let us write $\tau$ for the suite
$(p_1,\dots,p_{n+1},q_1,\dots,q_{n+1})$ of length $2(n+1)$; so $f$
realizes $\tau$.  Similarly, let us write $\tau'$ for the suite
consisting of the first $n+1$ points from the domain and from the
range of $\sigma'$.  Since $\sigma'$ belongs to PL, it is realized by
some projective linear transformation $f'$.  Of course this $f'$ also
realizes $\tau'$.  Furthermore, $f'$ sends the $j\th$ component $p'_j$
of $\sigma'$ to the corresponding component $q'_j$ in the range half
of $\sigma'$ (the $(l+j)\th$ component of $\sigma'$).

Now let $\sigma'$ vary, in PL, and approach the border suite
$\sigma$.  Then in particular, $\tau'$ approaches $\tau$, $p'_j$
approaches $p_j$, and $q'_j$ approaches $q_j$.  Applying
Lemma~\ref{continuity} (with $\tau'$ in the role of the $\sigma$ in the
lemma), we find that $f'(p'_j)$ approaches $f(p_j)$.  That is, $q'_j$
approaches $f(p_j)$.  But, since $\sigma'$ approaches $\sigma$, we
also know that $q'_j$ approaches $q_j$.  Therefore, $f(p_j)=q_j$, as
required.
\end{proof}

Now we turn to the single-qubit case where $\H = C^2$ and $\P$ is the
Riemann sphere $ = \CP^1$ that extends the field $\C$ of complex
numbers with an additional point $\infty$.

In this case, Claim~\ref{cla:range} simplifies somewhat, because
a tuple is in general position if and only if all its components are
distinct.  Indeed, since $n=2$ in this case, the definition of general
position requires simply that any two of the components are images, in
$\P$, of independent vectors in $\H$, which means that they are
distinct points in $\P$.

Our definition of ``suite'' requires the components in the
domain half to be distinct, so Claim~\ref{cla:range} has the
following consequence.

\begin{corollary}\label{cor:range}
In the single-qubit case, every suite on the border of PL has at
most two distinct points in its range half.
\end{corollary}

A point of $\P$ with
homogeneous coordinates $(a,b)$ can be conveniently represented as the
ratio $a/b$ where $a/b=\infty$  if $b=0$.

The Riemann-sphere representation of one-qubit states is closely related to the Bloch-sphere \cite{NC}, a representation of one-qubit states on the unit sphere $\s^2$ of the three dimensional Euclidean space $\E^3$. Recall the standard stereographic projection of $\s^2$ --- from the north pole onto the plane through the equator. Think of this plane as a copy of $\C$. Then the standard stereographic projection naturally extends to the stereographic projection of $\s^2$ onto the Riemann sphere by mapping the north pole onto $\infty$. Let $\pi$ be the inverse projection of the Riemann sphere onto the Bloch sphere. It is easy to check that $\pi\Q\ketpsi$ is the Bloch-sphere representation of the state given by the vector $\ketpsi$ in $\H$. In particular $\pi\infty$ is the north pole of $\s^2$. The Fubini-Study distance between points $z_1, z_2$ on the Riemann sphere is one half of the geodesic distance between the points $\pi z_1, \pi z_2$ on Bloch sphere.

\begin{lemma}\label{lem:rg}
If $L$ is a nonzero linear operator on $\H$ given by a matrix
$\begin{pmatrix} a &\; b\\ c &\; d \end{pmatrix}$
in the standard orthonormal basis of $\H$ then $\QL$ is the transformation
$z\mapsto \displaystyle\frac{az + b}{cz + d}\ \cdot$
\end{lemma}

\begin{proof}
  Let $\v = \alpha\e_1 + \beta\e_2$ and $z =
  \alpha/\beta$.  If $\beta \ne 0$ then we have
\begin{equation*}
  (\QL)z =\; \Q(L\begin{pmatrix}\alpha\\
                                \beta\end{pmatrix}\bigr)
             =\; \Q(L\begin{pmatrix} z     \\ 1 \end{pmatrix})
             =   \Q \begin{pmatrix} az+b  \\ cz+d \end{pmatrix}
             =  \frac{az+b}{cz+d}\, \cdot
\end{equation*}
If $\beta=0$ then $z = \infty$, and we have
\[
  (\QL)\infty =   \Q(L\begin{pmatrix}1 \\ 0\end{pmatrix})
             =   \Q \begin{pmatrix} a \\ c \end{pmatrix}
             =  \frac{a}{c} = \frac{a\infty + b}{c\infty + d}\, \cdot
\]
\end{proof}

\subsection{Cross-ratio}
\label{sub:xratio}

Projective linear transformations of the Riemann sphere are known as
fractional linear transformations and have the form
$\displaystyle\frac{c_1z + c_2}{c_3z + c_4}$ where $c_1,\dots,c_4$ are
complex numbers with $c_1c_4-c_2c_3\ne0$.

The cross-ratio of four distinct points $a,b,c,d$ on the Riemann
sphere is defined by
\[
\chi(a,b,c,d)=\frac{a-c}{b-c}\cdot\frac{b-d}{a-d}.
\]
If one of the four points is $\infty$, the cross-ratio is defined by
continuity; that amounts to just omitting those two of the four
factors that involve $\infty$. It is easy to check that the
cross-ratio is invariant under fractional linear transformations.

Although defined for tetrads of distinct points (general position),
the cross-ratio extends continuously to tetrads in which two of the
four points are equal while the other two are distinct (configuration
$2+1+1$), and also to tetrads in the configuration $2+2$, provided we
allow $\infty$ as a value for the cross-ratio. By
Corollary~\ref{cor:range}, configuration $2+1+1$ cannot occur in the
range of a suite on the border of PL, but configuration $2+2$ is
consistent with the corollary.

\begin{lemma}\label{lem:chi}
Let $a,b,c,d$ be points in the Riemann sphere $\P$.
\begin{enumerate}
\item If the tetrad $(a,b,c,d)$ is $2+2$ then $\chi(a,b,c,d) \in
  \{0,1,\infty\}$.
\item If $a,b,c,d$ are distinct then $\chi(a,b,c,d) \notin \{0,1,\infty\}$.
\end{enumerate}
\end{lemma}

\begin{proof}
To prove claim~(1), check that
\begin{itemize}
  \item if $a=b \ne c=d$ then $\chi(a,b,c,d)=1$,
  \item if $a=c \ne b=d$ then $\chi(a,b,c,d)=0$,
  \item if $a=d \ne b=c$ then $\chi(a,b,c,d)=\infty$.
\end{itemize}

We prove claim~(2) by reductio ad absurdum. Let $a,b,c,d$ be arbitrary
distinct points in $\P$.

First suppose that $\chi(a,b,c,d)\in\{0,\infty\}$.
If all points $a,b,c,d$ are complex numbers (not $\infty$) then
clearly $\chi(a,b,c,d) \notin \{0,\infty\}$. If $a=\infty$ then
$\chi(a,b,c,d) = \frac{b-d}{b-c} \notin \{0,\infty\}$. The cases
$b=\infty$, $c=\infty$ and $d=\infty$ are similar.

Second suppose that $\chi(a,b,c,d) = 1$. This is equivalent to each of
the following equations:
\begin{align*}
  (a-c)(b-d)&=(b-c)(a-d)\\
  ab-ad-bc+cd&=ab-bd-ac+cd\\
  bd+ac-ad-bc&=0\\
  (a-b)(c-d)&=0.
\end{align*}
So either $a=b$ or $c=d$.
\end{proof}

Every suite $\sigma$ of length 8 consists of a domain tetrad
$\dom{\sigma}$, where all four points are distinct, and a range tetrad
$\range{\sigma}$.

\begin{lemma}\label{lem:2+2}
  No suite $\sigma$ of length 8 such that $\range{\sigma}$ is $2+2$ is a
  limit point of the PL manifold of suites of length 8.
\end{lemma}

\begin{proof}
  By reductio ad absurdum, suppose that $\sigma$ is a suite of length
  8 such that $\range{\sigma}$ is $2+2$ and $\sigma$ is a limit point of
  the PL manifold of suites of length 8. Then there is a sequence
  $\tau_1, \tau_2, \dots$ of PL suites of length 8 that converges to
  $\sigma$. In particular, the domain tetrads $\dom{\tau_k}$ of suites
  $\tau_k$ converge to $\dom{\sigma}$, and the range tetrads $\range{\tau_k}$
  of suites $\tau_k$ converge to $\range{\sigma}$. By continuity,
  cross-ratios $\chi(\dom{\tau_k})\to \chi(\dom{\sigma})$ and cross-ratios
  $\chi(\range{\tau_k})\to \chi(\dom{\sigma})$ as $k\to\infty$.
  Since fractional linear transformations preserve cross-ratios, every
  $\chi(\dom{\tau_k}) = \chi(\range{\tau_k})$, and so
\begin{align*}
 \chi(\dom{\sigma}) &= \lim_{k\to\infty} \chi(\dom{\tau_k})\\
                 &= \lim_{k\to\infty} \chi(\range{\tau_k})\\
                 &= \chi(\range{\sigma})
\end{align*}
  which contradicts Lemma~\ref{lem:chi}
\end{proof}

In contrast to configurations $2+1+1$ and $2+2$, the cross-ratio does
not extend continuously to tetrads in the configurations $3+1$ or $4$.
Indeed, any neighborhood of any tetrad in either of these
configurations contains general-position tetrads with all possible
cross-ratios.

To see this, it suffices to prove the claim for one tetrad of each of
these two sorts, say $(0,0,0,\infty)$ and $(0,0,0,0)$.  (This
sufficiency follows immediately from the facts that the group of
fractional linear transformations acts transitively on each of these
two sorts of tetrads (because it acts doubly transitively\footnote{In
  fact it acts triply transitively, but that's not relevant here.} on
the Riemann sphere) and preserves cross-ratios.)  Given any possible
cross-ratio, we can find two tetrads $(a,b,c,d)$ and
$(a',b',c',\infty)$ with that cross-ratio, where all of
$a,b,c,d,a',b',c'$ are complex numbers (not $\infty$). Now apply to
these tetrads the fractional linear transformation $z\mapsto\eps z$
for a very small, positive, real $\eps$.  The resulting tetrads have
the same cross-ratio and are very close --- arbitrarily close as
$\eps\to0$ --- to $(0,0,0,0)$ and $(0,0,0,\infty)$, respectively, as
claimed.

\subsection{Proof of Infinite Approximability Theorem}

We assume that $\ell\ge3$ and the principal quantum system consists of
just one qubit. Fix an orthonormal basis $\e_1,\e_2$ in $\H = \C^2$.

\begin{lemma}\label{lem:singular}
  A suite $\sigma$ is exactly realizable by means of a nonzero
  singular linear operator $L$ on $\H$ (so that $\sigma$ extends to
  $\QL$) if and only if all $\ell$ points in $\range{\sigma}$ are equal.
\end{lemma}

\begin{proof}
  If a nonzero singular linear operator $L$ exactly realizes the given
  suite $\sigma$ then $L\e_1$ and $L\e_2$ are collinear and $\QL$ is
  constant. Hence all points in $\range{\sigma}$ are equal.

  If all points in $\range{\sigma}$ are equal, say to a point $\Q\v$, then
  the desired $L$ can be obtained by setting $L\e_1 = L\e_2 = \v$.
\end{proof}

\begin{proposition}[Border of PL]
\label{pro:border}
  A suite $\sigma$ belongs to the border of the PL manifold if and
  only if the range part $\range{\sigma}$ satisfies one of the following
  two conditions.
  \begin{enumerate}
    \item All $\ell$ points in $\range{\sigma}$ are equal.
    \item Exactly $\ell-1$ of the $\ell$ points in $\range{\sigma}$ are
      equal.
  \end{enumerate}
\end{proposition}

\begin{proof}
  We first prove the only-if implication. Suppose that $\sigma$
  belongs to the border of PL. By Claim~\ref{cla:range}, the range
  tuple $\range{\sigma}$ contains at most two distinct points. If all
  points in $\range{\sigma}$ are equal, we are done. Suppose that
  $\range{\sigma}$ contains exactly two distinct points. Then the index set
  $\{\ell+1, \dots, 2\ell\}$ splits into disjoint parts $I$ and $J$
  such that the same point $p$ occurs in all $I$ positions and a
  different point $q$ occurs in all $J$ positions. Without loss of
  generality, $I$ contains at least two indices. Suppose toward a
  contradiction that $J$ contains at least two indices as well, so
  that $\ell\ge4$. Then there is a suite $\sigma_0$ of length 8
  embedded in the suite $\sigma$ of length $2\ell$ such that the range
  of $\sigma_0$ is of type $2+2$ and $\sigma_0$ is a limit point of
  the PL manifold of suites of length 8. This contradicts
  Lemma~\ref{lem:2+2}.

  Next we prove the if implication. Notice that, in either of the two
  cases, $\sigma$ does not belong to PL. Indeed, fractional linear
  transformations (and projective linear transformations in general)
  preserve equality and disequality, and so all points in the range
  part of a PL suite are distinct. It remains to prove that $\sigma$
  belongs to the closure of PL.

  If all points in $\range{\sigma}$ are equal then, by the preceding lemma,
  $\sigma$ is exactly realizable by means of a nonzero singular linear
  operator $L$. Now use Lemma~\ref{lem:approximants}.

  Suppose that some $\ell-1$ points in $\range{\sigma}$ are equal, say to a
  point $p$, but another point $q$ also occurs in $\range{\sigma}$. It
  suffices to consider the case where $p=0$ and $q= \infty$. Indeed,
  there is a fractional linear transformation $f$ that moves $p,q$ to
  $0,\infty$ respectively. If PL suites $\tau_k$ converge to $f\sigma$
  then PL suites $f^{-1}\tau_k$ converge to $\sigma$.

  Without loss of generality, $\infty$ occurs in the very last
  position in $\sigma$, so that $\sigma$ has the form
\[
  (p_1, \dots, p_{\ell-1}, p_\ell, 0, \dots,0, \infty).
\]
  There is a fractional linear transformation $g$ that sends $p_\ell$
  to $\infty$. For each $k=1,2,\dots$, let $g_k$ be the fractional
  linear transformation $g/k$, and let $\tau_k$ be the restriction of
  $g_k$ to $\dom{\sigma}$. The sequence of PL suites $\tau_k$ converges
  to $\sigma$.
\end{proof}

Theorem~\ref{thm:inf} follows from Proposition~\ref{pro:border} and
Lemma~\ref{lem:singular}.

\section{Final remarks}
\label{sec:final}

\subsection{Mixed states}

We have been working with pure states. One may consider a
generalization of the results above to mixed states. Here we just
point out that the scenario in the beginning of our story readily
generalizes to mixed states and channel representation.

As is, the scenario is not a channel; since the unsuccessful
measurement result is discarded, the trace is not preserved. But the
scenario becomes a channel if the unsuccessful measurement result is
not discarded and the measurement result is not looked at. The
modified scenario corresponds to a composition of three channels, as
follows. To simplify the exposition, we consider only the case of one
princpal qubit and one ancilla, and we presume that the designated
initial and final states of the ancilla are $\ket0$.

First, we take the essential qubit, in some state $\ketpsi$ and adjoin
to it a prepared ancilla.  This is, to begin with, a linear embedding
$E$ of the Hilbert space $\C^2$ for a single qubit into the Hilbert
space $\C^4$ for two qubits; it sends $\ketpsi$ to $\ket{0,\psi}$.
The embedding preserves lengths of vectors, and we get a channel by
sending each linear operator $R$ on $\C^2$ to the linear operator  $E
R E^\dag$ on $\C^4$.  This is the channel for the first part of our
scenario.

Second, we apply the unitary operator $U$ to the two-qubit system.
That corresponds to the channel that sends any linear operator $R$ on
$\C^4$ to $U R U^\dag$.

Finally, we measure the ancilla.  The measurement involves two
projection operators $P_0$ and $P_1$, from $\C^4$ to $\C^2$,
corresponding to the values $\ket0$ and $\ket1$ for the ancilla
respectively.  Each $P_i$ induces a linear transformation $T_i(R) =
P_i R P_i^\dag$ from linear operators on $\C^4$ to linear operators on
$\C^2$, but the projection operators $P_i$ do not preserve lengths of
vectors, and the transformations $T_i$ do not preserve traces. If we
discarded the results in the case of failure, we'd have only $T_0$,
which isn't a channel.  But by keeping the (one-qubit) result in both
cases, we get
$T_0+T_1: R \mapsto\sum_{i=0}^1 P_i R P_i^\dag$,
and this is a channel, with
$\sum_{i=0}^1 P_i P_i^\dag = I + I = 2I.$
The factor 2 is needed to make the superoperator trace-preserving; the
source and target Hilbert spaces have different dimensions.

Composing the three parts, we have a channel
\[
  R \mapsto\sum_{i=0}^1 (P_i U E) R (P_i U E)^\dag.
\]

\subsection{One numerical function of suites}

If $L$ is a nonzero linear operator on $\H$, let $\rho(L)$ be the
ratio $\lambda_{\min}/\lambda_{\max}$ of the minimal and maximal
eigenvalues of $L^\dag L$. According to Theorem~\ref{thm:exact},
$\rho(L)$ is the maximum of the guaranteed success probabilities of
unitary operators on $\H^+$ realizing $L$ exactly. Since $\rho(L) =
\rho(L')$ if if $L,L'$ are collinear, define $\rho(\QL) =
\rho(L)$. For any exactly realizable suite $\sigma$ define
$\rho(\sigma)$ to be the supremum of $\rho(\QL)$ taken over all
nonzero linear operators $L$ such that $\sigma$ extends to
$\QL$. Finally, observe that a nonzero linear operator $L$ is singular
if and only if $\rho(L) = 0$.

\begin{claim}
  If $\sigma$ is a suite in the closure of PL such that, for some
  fixed $\rho_0>0$, every neighborhood of $\sigma$ contains a PL suite
  $\tau$ with $\rho(\tau) > \rho_0$, then $\sigma$ belongs to PL.
\end{claim}

An equivalent way to formulate the claim is that, if $\sigma$ is on
the border of PL then every sequence $\tau_1, \tau_2, \dots$ of PL
suites converging to $\sigma$ must have $\rho(\tau_k)\to0$ as
$k\to\infty$.

\begin{proof}
  Apply the hypothesis of the claim to choose a sequence $\tau_1,
  \tau_2, \dots$ of PL suites converging to
  $\sigma$ and having $\rho(\tau_k) > \rho_0$. Choose linear
  operators $L_k$ with $\rho(L_k) > \rho_0$ such that $\tau_k$ extends
  to $\QL_k$; we may assume that $\lambda_{\max}(L_k) = 1$, so that
  $\lambda_{\min}(L_k) > \rho_0$.

  Every $L_k$ has a polar decomposition $U_k P_k$ where $U_k$ is
  unitary and $P_k$ is Hermitian and positive definite (not just
  semi-definite, because $L_k$ is invertible).  Notice that $L_k^\dag
  L_k = P_k^\dag P_k$.  Since $P_k$ is Hermitian, it can be
  diagonalized, say $P_k = B_k D_k B_k^\dag$ where $B_k$ is unitary
  and $D_k$ is diagonal.  Since the eigenvalues of $P_k^\dag P_k$ lie
  between $\rho_0$ and 1, the diagonal entries in $D$ lie between
  $\sqrt{\rho_0}$ and 1.

  By passing to a subsequence of $\tau_1, \tau_2, \dots$, we can
  arrange that the unitary matrices $U_k$ converge to some unitary
  matrix $U$ (because the unitary group is compact), that the unitary
  matrices $B_k$ converge to a unitary matrix $B$ (same reason), that
  therefore $B_k^\dag\to B^\dag$, that the diagonal matrices $D_k$
  converge to a diagonal matrix $D$ (because the eigenvalues all lie
  in the bounded interval $[\sqrt{\rho_0},1]$), and that therefore the
  matrices $P_k$ converge to a matrix $P = BDB^\dag$, and the $L_k$
  converge to some $L = U P$.

  Because of the convergence, we have that the eigenvalues of $D$ lie
  in $[\sqrt{\rho_0},1]$ and, in particular, are positive.  So $D$ is
  invertible, and therefore so are $P = BDB^\dag$ and $L = UP$.

  Finally, let $\vec d$ and $\vec r$ be the domain and range parts of
  $\sigma$ respectively, and let $\vec d_k$ and $\vec r_k$ be the
  domain and range parts of $\tau_k$ respectively. Since $L_k(\vec
  d_k) = \tau_k(\vec d_k) = \vec r_k\to \vec r$ as $k\to\infty$, and
  as $L_k(\vec d)\to L(\vec d)$ by continuity, we have $L(\vec d) =
  \vec r$.  Thus, $\sigma$ is in PL, as claimed.
\end{proof}

\subsection{Inapproximability in the single-qubit case}

According to \S\ref{sub:xratio}, the cross-ratio does not extend
continuously to tetrads of points in Riemann sphere that are in
configurations $3+1$ or $4$; any neighborhood of any tetrad in either
of these configurations contains general-position tetrads with all
possible cross-ratios. Thus there cannot be a theorem of the form: If
the cross-ratios of two general-position tetrads differ by at least
$\eps$, then the suite of length $8$ consisting of these two tetrads
cannot be within $\delta$ of the FL manifold of suites of length
$8$. Indeed, no matter how big we make $\eps$ and how small we make
$\delta$, counterexamples can be found within $\delta$ of the
double-suite $(0,0,0,0;0,0,0,0)$.

The best we can hope to do in the direction of such an
inapproximability theorem is to assume, as an additional hypothesis,
that the tetrads in question are bounded away from the singular locus
of the cross-ratio, i.e., the locus $S$ of tetrads of configurations
$3+1$ and $4$.

\begin{claim}
  Let $\eps$ and $\gamma$ be positive real numbers.  Then there is a
  positive real $\delta$ with the following property.  Let $t$ and
  $t'$ be tetrads whose distance from the singular locus $S$ of the
  cross-ratio function is at least $\gamma$.  Suppose further that the
  distance between their cross-ratios is at least $\eps$.  Then the
  distance between $t$ and $t'$ is at least $\delta$.
\end{claim}

\begin{proof}
  Let $\gamma>0$ be given and let $D$ be the space of tetrads whose
  distance from $S$ is at least $\gamma$.  This is a closed subsace of
  the compact space of all tetrads, so it is also compact.  The
  cross-ratio is a continuous function from $D$ to the Riemann sphere,
  so, by compactness, it is uniformly continuous.  Given $\eps>0$, let
  $\delta>0$ be as in the definition of uniform continuity: Any two
  points of $D$ whose distance is $<\delta$ have cross-ratios whose
  distance is $<\eps$.  In view of the invariance of the
  cross-ratio under fractional linear transformations, that is exactly
  (the contrapositive of) the assertion
  of the proposition.
\end{proof}

  The preceding argument is valid for any distance functions inducing
  the usual topologies on the space of tetrads and on the Riemann
  sphere. Quantitative information about how the $\delta$ in the claim
  varies as a function of $\gamma$ and $\eps$ could be obtained by
  methods of elementary calculus (Lagrange multipliers).

\subsection{Variety for PL in the single-qubit case}
\label{sub:variety}

Proposition~\ref{pro:variety} asserts that, in any coordinate patch of
$\P^{2\ell}$, there exists an algebraic variety $V$ such that $\PL
\subseteq V$ and $G \cap V \subseteq \PL$ for some full-measure open
set $G$. The proof of the proposition is not constructive. We can do
better and provide a constructive proof for the proposition. Here we restrict attention to the single qubit case where the construction is especially easy due to the cross-ratio function. In the general case, one can use the construction of proof of Lemma~\ref{lem:n+1}.

If $\ell \le 3$ then, by Proposition~\ref{lem:n+1}, every
general-position suite belongs to PL, so $V$ could be given by $0=0$.
Suppose that $\ell\ge4$.

Let $V$ be the algebraic variety in the Riemann sphere, in variables
$a_1,\dots,a_\ell, b_1,\dots,b_\ell$ given by $\ell-3$ polynomial
equations obtained from $\ell-3$ equations
\begin{equation}\label{v}
  \chi(a_1,a_2,a_3,a_i) = \chi(b_1,b_2,b_3,b_i)\quad\
  \text{where}\ i = 4, \dots, \ell
\end{equation}
by clearing fractions.

\begin{claim}\label{cla:variety1}
  $\PL\subseteq V$ and every general-position suite in $V$ belongs to
  PL.
\end{claim}

\begin{proof}
The inclusion $\PL\subseteq V$ follows from the fact that fractional
linear transformations preserve cross-ratios. Suppose that a
general-position suite $\sigma$ satisfies our polynomial equations.
Then it also satisfies the equations \eqref{v}. By
Lemma~\ref{lem:n+1}, there is a fractional linear transformation $f$
that sends $(a_1,a_2,a_3)$ to $(b_1,b_2,b_3)$. For each
$i=4,\dots,\ell$, we have also
\begin{align*}
  \chi(b_1,b_2,b_3,b_i) &= \chi(a_1,a_2,a_3,a_i)\\
                        &= \chi(fa_1,fa_2,fa_3,fa_i)\\
                        &= \chi(b_1,b_2,b_3,fa_i),
\end{align*}
which implies that $f(a_i)=b_i$.
\end{proof}

Finally, let's consider suites of length 8 with domain
$(0,\infty,1,-1)$.

\begin{claim}
Every PL suite of the form $(0,\infty,1,-1, a,b,c,d)$ satisfies the
equation
\begin{equation}\label{averages}
 \frac{ab+cd}{2} = \frac{a+b}{2} \cdot \frac{c+d}{2}
\end{equation}
\end{claim}

Equation \eqref{averages} is easy to remember due to the slogan
``average of products equals product of averages.''

\begin{proof}
First we show that every PL suite of the form $(0,\infty,1,q,
a,b,c,d)$ satisfies the equation
  \begin{equation}\label{q}
   (1-q)(ab+cd) = a(c-qd) + b(d-qc).
  \end{equation}
where $q$ is the inverse of $\chi(a,b,c,d)$.

Indeed, the unique fractional linear transformation sending $(a,b,c)$
to $(0,\infty,1)$ is
\[
 z \mapsto \frac{z-a}{z-b}\cdot \frac{c-b}{c-a}
\]
So $(a,b,c,d)$ is the image of $(0,\infty,1,q)$ under a fractional
linear transformation if and only if the exhibited transformation
sends $d$ to $q$, i.e., if and only if
\begin{equation}\label{q2}
 q = \frac{d-a}{d-b} \cdot \frac{c-b}{c-a}
\end{equation}
so that $q$ is the inverse of $\chi(a,b,c,d)$.
Clearing fractions and rearraging terms in \eqref{q2}, we get equation
\eqref{q} which yields equation \eqref{averages} in case $q=-1$.
\end{proof}

\end{document}